\documentclass[journal, onecolumn,11pt]{IEEEtran}
\IEEEoverridecommandlockouts

\usepackage[switch]{lineno}
\usepackage[named]{algo}
\usepackage[noend]{algpseudocode}

\usepackage{algorithm}
\usepackage{amsthm}
\usepackage{amsfonts}
\usepackage{amssymb}
\usepackage{mathrsfs}
\usepackage{mathtools}
\usepackage[english]{babel}
\usepackage{float}
\usepackage[pdftex]{graphicx}
\usepackage[dvipsnames]{xcolor}
\usepackage{epstopdf}
\usepackage{amsmath}
\usepackage{color}
\usepackage{multirow}
\usepackage{bm, tikz}  
\usepackage{multirow,booktabs}
\usepackage{graphicx}
\usepackage{bm}
\usepackage{balance}
\usepackage{microtype}
\usepackage[caption=false,font=footnotesize]{subfig}
\usepackage{cite}
\usepackage{comment}
\allowdisplaybreaks
\usepackage{algorithm,algorithmicx,algpseudocode}
\usepackage{caption}
\usepackage{graphbox}
\usepackage{placeins}
\usepackage{wrapfig}
\usepackage{subcaption}
\usepackage{etoolbox}

\newtoggle{hqfigures}
\togglefalse{hqfigures}  

\DeclareMathOperator*{\argmin}{arg\,min }

\newtheorem{theorem}{Theorem}

\newtheorem{remark}{Remark}
\usepackage[most]{tcolorbox}

\newcommand{\boxedthm}[1]{
\begin{tcolorbox}[colback=green!10,
                  colframe=green!10,
                  width=\linewidth,
                  arc=1mm, auto outer arc,
                  boxrule=0pt,
                  boxsep=-1mm,
                 ]
  #1
\end{tcolorbox}}

\usepackage{amsmath}
\usepackage{graphicx}
\usepackage[colorlinks=true, allcolors=blue]{hyperref}

\definecolor{customgreen}{HTML}{D9FFD9}
\definecolor{customblue}{HTML}{D9D9FF}
\definecolor{customteal}{HTML}{BFDFDF}
\definecolor{customorange}{HTML}{FFDFBF}
\newcommand\blfootnote[1]{%
  \begingroup
  \renewcommand\thefootnote{}\footnotetext{#1}%
  \addtocounter{footnote}{-1}%
  \endgroup
}

\usepackage{ulem}
\usepackage{mathrsfs}

\usetikzlibrary{arrows.meta}

\begin{document}

\title{UTOPY: Unrolling Algorithm Learning via Fidelity Homotopy for Inverse Problems \thanks{This work was supported by the Army Research Office/Laboratory under grant number W911NF-25-1-0165, VIE from UIS project 8087. The views and conclusions contained in this document are those of the authors and should not be interpreted as representing the official policies, either expressed or implied, of the Army Research Laboratory or the U.S. Government.}
\\
}

\author{Roman Jacome$^\dag$,  Romario Gualdr\'on-Hurtado$^\ddag$, Leon Suarez-Rodriguez$^\ddag$ and Henry Arguello$^\ddag$\\
{$^\dag${Department of Electrical, Electronics, and Telecommunications Engineering}} \\{$^\ddag${Department of Systems Engineering and Informatics}} \\
 {Universidad Industrial de Santander, Colombia, 680002}
 
}

\maketitle
\blfootnote{\footnotesize \textcopyright~ 2025 IEEE. Personal use of this material is permitted. Permission from IEEE must be obtained for all other uses, in any current or future media, including reprinting/republishing this material for advertising or promotional purposes, creating new collective works, for resale or redistribution to servers or lists, or reuse of any copyrighted component of this work in other works. }
\begin{abstract}

Imaging Inverse problems aim to reconstruct an underlying image from undersampled, coded, and noisy observations. Within the wide range of reconstruction frameworks, the unrolling algorithm is one of the most popular due to the synergistic integration of traditional model-based reconstruction methods and modern neural networks, providing an interpretable and highly accurate reconstruction. However, when the sensing operator is highly ill-posed, gradient steps on the data-fidelity term can hinder convergence and degrade reconstruction quality. To address this issue, we propose {UTOPY}, a homotopy continuation formulation for training the unrolling algorithm. Mainly, this method involves using a well-posed (synthetic) sensing matrix at the beginning of the unrolling network optimization. We define a continuation path strategy to transition smoothly from the synthetic fidelity to the desired ill-posed problem.  This strategy enables the network to progressively transition from a simpler, well-posed inverse problem to the more challenging target scenario. We theoretically show that, for projected gradient descent-like unrolling models, the proposed continuation strategy generates a smooth path of unrolling solutions. Experiments on compressive sensing and image deblurring demonstrate that our method consistently surpasses conventional unrolled training, achieving up to 2.5 dB PSNR improvement in reconstruction performance. Source code at \texttt{\url{https://github.com/romanjacome99/utopy}}

\end{abstract}

\begin{IEEEkeywords}
Inverse problems, Unrolling algorithm, homotopy optimization, compressed sensing, image deblurring. 
\end{IEEEkeywords}
\section{Introduction}
\vspace{-0.1cm}
Inverse problems involve reconstructing an unknown signal from noisy, corrupted, or typically undersampled observations, making the recovery process generally non-invertible and ill-posed. This work focuses on linear inverse problems, where a sensing matrix represents the forward model \cite{bertero1985linear}. Numerous imaging tasks rely on these principles, including image restoration—such as deblurring, denoising, inpainting, and super-resolution \cite{gualdron2024learning,gunturk2018image,yang2010image}—as well as compressed sensing (CS) \cite{cs} and medical imaging applications like magnetic resonance imaging (MRI) \cite{lustig2008compressed} or computed tomography (CT) \cite{willemink2013iterative}. See  \cite{IPDL,bai2020deep,bertero2021introduction} and references therein for more applications.

Algorithm unrolling offers an interpretable and effective framework for solving inverse problems by embedding neural network architectures within classical optimization procedures \cite{unrolling_review,gregor2010learning}. Unrolled models replace fixed iterative schemes with learned proximal operators and data-driven parameters, yielding improved accuracy and computational efficiency in applications such as image restoration \cite{mou2022deep}, MRI recovery \cite{aggarwal2018modl}, snapshot spectral imaging \cite{wang2019hyperspectral}, graph signal processing \cite{li2022graph}, and phase retrieval \cite{pinilla2023unfolding}, among others. Unrolling networks fall within the framework of learning-to-optimize \cite{chen2022learning}, which automates the design of optimization methods based on their performance on a set of training problems.  Traditional unrolling approaches optimize against a fixed data fidelity objective, leading to convergence difficulties with highly ill-posed measurement operators.

We propose UTOPY, a homotopy continuation formulation for unrolling algorithm optimization to address ill-posed data fidelity within the training. Homotopy continuation methods \cite{allgower2003introduction} offer a general strategy to mitigate such challenges by gradually transforming an easier optimization problem into the target problem. Related ideas appear in curriculum learning \cite{bengio2009curriculum} and progressive training for generative models \cite{karras2018progressive}, where tasks are ordered to improve training dynamics. Recent work on learned proximal models \cite{adler2018learned} focuses on regularization design but does not incorporate a progressive fidelity path. In inverse problems, $\ell_1$ -homotopy has been widely studied for CS, where a sequence of decreasing soft/hard thresholding parameters is used to approximate the sparse solution \cite{zhang2015simple,dong2017homotopy}.
Homotopy-based methods have also been explored in physics-informed neural networks \cite{zheng2023hompinns}, which leverage homotopy continuation to navigate multiple solution branches in nonlinear partial differential equations inverse problems. While conceptually related, such approaches have not yet been extended to algorithm unrolling for general inverse imaging problems.

Our method systematically transitions from a simpler, well-posed problem towards the challenging target fidelity, enhancing performance and generalization capability. This progressive strategy starts by learning the image prior for an easier inverse problem and by smoothly transitioning to the desired inverse problem, improving solution accuracy and computational robustness. We provide theoretical guarantees on the smooth transition of the different solutions in the continuation path. The experiments show the performance of the proposed method for single-pixel imaging and image deblurring, validating the proposed framework. Additionally, this training scheme also has generalization benefits to changes in the sensing model at testing time, as the model was trained with a linear combination of the synthetic inverse problem and the real one. Although the method was validated with the FISTA unrolling formulation, it can be extended to other approaches, such as ADMM or HQS formulations \cite{unrolling_review} and also deep equilibrium architectures \cite{gilton2021deep}. \vspace{-0.35cm}
\section{Background}

 Consider the inverse problem of reconstructing an image \(\mathbf{x} \in \mathbb{R}^n\) from noisy and undersampled measurement $\mathbf{y} = \mathbf{H}\mathbf{x} + \mathbf{e}\in\mathbb{R}^m
$ with \(m \leq n\), \(\mathbf{H} \in \mathbb{R}^{m\times n}\) is the sensing matrix characterizing the {acquisition} system, and \(\mathbf{e} \in \mathbb{R}^{m}\) is additive sensor noise. To estimate the original image \(\mathbf{x}\) from measurements \(\mathbf{y}\),  variational formulations aim to solve the following optimization problem:
\vspace{-0.2cm}
\begin{equation}
\vspace{-0.1cm}
\label{eq:unfolding_optimization}
    \hat{\mathbf{x}} = \argmin_{\mathbf{x}} g(\mathbf{x}) + \lambda h(\mathbf{x}),
\end{equation}
where \(g(\mathbf{x})\) is the data fidelity term, and \(h(\mathbf{x})\) acts as a regularizer, embedding prior knowledge about the image structure. The parameter \(\lambda > 0\) controls the strength of the regularization. Unfolding methods integrate iterative optimization algorithms directly into neural network architectures, where each iteration corresponds to a neural network layer with trainable parameters \cite{unrolling_review}. The proximal step in classical methods is replaced by a neural network to learn task-specific image priors from data effectively.  Without loss of generality, we consider projected gradient descent-based unrolling with the fidelity term \(g(\mathbf{x})=\frac{1}{2}\|\mathbf{y}-\mathbf{H}\mathbf{x}\|_2^2\), the unfolding step at iteration \(k\) can be formulated as:
\vspace{-0.2cm}
\begin{equation}
\label{eq:unfolding_step}
    \small{\mathbf{x}^{k} = \mathcal{D}_{\theta^k}\left(\mathbf{x}^{k-1} - \tau^k \overbrace{\mathbf{H}^\top(\mathbf{H}\mathbf{x}^{k-1}-\mathbf{y})}^{\nabla g(\mathbf{x}^{k-1})}\right) \coloneqq T_{\Omega^k}(\mathbf{x}^k)},\vspace{-0.35cm}
\end{equation}
where \(\mathcal{D}_{\theta^k}\) represents a neural network parameterized by \(\theta^k\), \(\tau^k > 0\) denotes the $k-th$ step size with $K$ is the {maximum} number of iterations/stages and $\Omega^k = \{\theta^k, \tau^k\}$ are the $k-th$ trainable parameters of the unfolding algorithm with $\mathbf{x}^0 = \mathbf{H}^\top \mathbf{y}$. In practice, we add an acceleration step, which improves convergence of the algorithm as the one in FISTA \cite{FISTA}. The unfolding algorithm is trained as follows:
\vspace{-0.3cm}
\begin{equation}
\label{eq:training_problem}
\hat{\Omega} = \argmin_{\Omega} \sum_{p=1}^{P}\mathcal{L}_{\texttt{rec}}\left(\mathbf{x}_p , T_{\Omega^K}(\dots T_{\Omega^1}(\mathbf{H}^\top\mathbf{y}_p))\right),\vspace{-0.3cm}
\end{equation}
where $\mathcal{L}_{\texttt{rec}}(\cdot,\cdot)$ is a reconstruction loss function such as MSE or $\ell_p$-norm, \(\mathbf{x}_p\) and \(\mathbf{y}_p\) are the training pairs of ground-truth images and measurements, respectively. Training these unfolding architectures involves learning the parameters \({\mathbf{\Omega}} = \{\Omega_1, \dots, \Omega_K\}\), enabling the model to implicitly capture complex, data-driven priors tailored for specific inverse problems.
\section{Fidelity Homotopy for Unrolling Algorithm}
Here we describe UTOPY, a progressive fidelity approach that leverages a homotopy continuation method to effectively train unrolling algorithms by smoothly transitioning from an easier fidelity constraint to the challenging fidelity constraint. 
\subsection{Homotopy Design}\vspace{-0.2cm}
First, let denote $\mathbf{y}_t = \mathbf{H}_t\mathbf{x} \in \mathbb{R}^m_t$, $m_t\geq m$,  a synthetic measurement set (used only during training) which is related to the real model but is better-posed. We consider two $\mathbf{H}_t$ designs for the applications:

\textbf{Compressed sensing} Here we consider the case where $\mathbf{H}= \mathbf{PW}$, where $\mathbf{W}\in\mathbb{R}^{n\times n}$ is an orthogonal transformation e.g., Hadamard or Fourier, and $\mathbf{P}\in\mathbb{R}^{m\times n}$ is a subsampling and permutation matrix that selects few rows of the transformation matrix. Thus, we consider $\mathbf{H}_t = \mathbf{P}_t\mathbf{W}$ where $\mathbf{P}_t \in \mathbb{R}^{m_t\times n}$, selects more rows from $\mathbf{W}$ than $\mathbf{P}$. In this setting, we define the hyperparameter $\eta = 1-m/n$ as the augmented ratio.

\textbf{Image Deblurring} The forward model $\mathbf{H}$ is built upon a Toeplitz matrix based on the convolution kernel denoted as  $\mathbf{H}[i,i+j] = \mathbf{h}[j]$  with $i=1,\dots, m$ and $j=1,\dots,n$, where $\mathbf{h}$ is the vectorized Gaussian kernel with $\sigma$ standard deviation. Thus, we consider $\mathbf{H}_t[i,i+j] = \mathbf{h}_t[j]$ with the kernel $ \mathbf{h}_t$ has a standard deviation $\sigma_t < \sigma$. Note that in this case $m_t =m$, which depends on the employed padding scheme. The implementation contains zero-padding such that $n=m_t=m$.

Define the family of fidelity functions $g_\alpha(\mathbf{x)} = (1 - \alpha) \frac{1}{2}\| \mathbf{y} - \mathbf{H}\mathbf{x} \|^2_2 + \alpha  \frac{1}{2}\| \mathbf{y}_t - \mathbf{H}_t \mathbf{x} \|^2_2$, where \(\alpha\) is a homotopy parameter controlling the transition from a simpler fidelity \((\mathbf{H}_t, \mathbf{y}_t)\) towards the target fidelity \((\mathbf{H}, \mathbf{y})\). The continuation path is designed so that $\alpha: 1 \rightarrow 0$ during the training. Then, the fidelity homotopy optimization is expressed as follows:
\vspace{-0.1cm}
\begin{equation}
\vspace{-0.1cm}
\label{eq:fidelity_homotopy}
\hat{\mathbf{x}}_\alpha= \arg\min_{\mathbf{x}} g_\alpha(\mathbf{x})+ \lambda h(\mathbf{x}),
\end{equation}
then, the unrolling network for this optimization problem becomes\vspace{-0.3cm}
\begin{equation}
    \mathbf{x}_\alpha^k = \mathcal{D}_{\theta^k}\left(\mathbf{x}_\alpha^{k-1} - \tau^k {\nabla g_\alpha(\mathbf{x}_\alpha^{k-1})}\right), \text{ where }
\end{equation}
\begin{equation}
   \small{ {\nabla g_\alpha(\mathbf{x}) = \frac{1}{2}\alpha\mathbf{H}_t^\top (\mathbf{H}_t\mathbf{x} - \mathbf{y}_t) + \frac{1}{2}(1-\alpha)\mathbf{H}^\top(\mathbf{Hx} -\mathbf{y})}.}\vspace{-0.3cm}
\end{equation}
\vspace{-0.5cm}
\subsection{ Fidelity Homotopy for Unrolling Algorithm Optimization}
\vspace{-0.1cm}
\begin{algorithm}[!t]
\caption{Fidelity Homotopy for Unrolling Algorithm }
\label{alg:fidelity_optimization}
\begin{algorithmic}[1]\small{
\Require \(\mathbf{H}_t,\mathbf{H}\), \(\mathbf{y}_t,\mathbf{y}\), \(\texttt{max\_epochs}, \epsilon, \texttt{freq}\)
\State Initialize network weights \({\mathbf{\Omega}} = \{\Omega_1,\dots,\Omega_K\}\) and set \(\alpha = 1\)
\For{\(\ell = 1:\texttt{max\_epochs}\)}
    \State if \((\ell \mod{\texttt{freq}} = 0)\), then update \(\alpha = \texttt{scheduler}_\epsilon(\ell),\) \Comment{Update $\alpha$ according scheduling}
    \State \(\mathbf{x}_\alpha^0 = \mathbf{z}_\alpha = \alpha\mathbf{H}_t^\top\mathbf{y}_t + (1-\alpha)\mathbf{H}^\top \mathbf{y}\) \Comment{Algorithm initialization}
    \For{\(k=1:K\)}
        \State $\mathbf{x}^{k+1}_\alpha =\mathcal{D}_{\theta^k}\left(\mathbf{z}_\alpha - \tau_k \nabla_\mathbf{x} g_{\alpha}(\mathbf{z}_\alpha)\right)$ \Comment{Projected gradient descent step}

        \State $\mathbf{z}_\alpha = \mathbf{x}^{k+1}_\alpha  - t^k( \mathbf{x}^{k+1}_\alpha -  \mathbf{x}^{k}_\alpha)$ \Comment{Acceleration Step}
    \EndFor
    \State Compute loss \(\mathcal{L} = \mathcal{L}_{\texttt{rec}}( \mathbf{x},{\mathbf{x}}_\alpha^K)\)
    \State Update weights \(\Omega = \Omega - \zeta \nabla_\Omega \mathcal{L}\)
\EndFor
}\vspace{-0.3cm}
\boxedthm{\textbf{Note:} During inference (testing), set \(\alpha = 0\) exclusively.}
\end{algorithmic}
\end{algorithm}

{Algorithm~\ref{alg:fidelity_optimization} details the fidelity homotopy, starting with network parameters initialization}. {Parameter} $\alpha$ is initialized to $1$, indicating the initial focus is on the simpler fidelity term ($\mathbf{H}_t$, $\mathbf{y}_t$), with a better-posed measurement operator. At each epoch $\ell$, the value of $\alpha$ is updated following $\texttt{scheduler}_\epsilon(\ell)$.  which are two: 
{\begin{equation}
    \overset{\text{Exponential}}{\texttt{scheduler}_\epsilon(\ell) = e^{-\epsilon \ell}},\; \overset{\text{Linear}}{
   \texttt{scheduler}_\epsilon(\ell) = 1-\ell/\epsilon},
\end{equation}}
where $\epsilon>0$ is a parameter controlling the speed of the decreasing values of $\alpha$. The exponential scheduling allows faster neural network convergence, but the constant decrease of the linear scheduling allows better performance as shown in Section \ref{sec:results} The initial estimate for the image at each epoch, $\mathbf{x}_\alpha^0$, is computed as a convex combination of adjoint reconstructions from both fidelity terms: $\mathbf{x}^0_\alpha = \alpha \mathbf{H}_t^\top \mathbf{y}_t + (1 - \alpha) \mathbf{H}^\top \mathbf{y}.$  In the inner loop, the algorithm updates the reconstruction iteratively through learned proximal operations. Finally, during inference, the parameter $\alpha$ is fixed to $0$, ensuring the neural network only uses the target fidelity.

\vspace{-0.2cm}
\subsection{Theoretical Analysis}
\label{sec:theory_analysis}
\vspace{-0.2cm}
\noindent We make the following assumptions:
(\textbf{A1}) ${\nabla g_{\alpha}}$ is $L_{\alpha}$-Lipschitz and continuously differentiable in $(\mathbf{x},\alpha)$ with $L=\sup_{\alpha}L_{\alpha}<\infty$.   
(\textbf{A2a}) 
The proximal network $\mathcal{D}_{\theta}\!:\mathbb{R}^{n}\!\to\!\mathbb{R}^{n}$ is built from affine layers and pointwise $\mathcal{C}^{1}$ activations (e.g.\ Swish, GELU, soft-ReLU).  The chain rule implies $\mathcal{D}_{\theta}$ itself is differentiable everywhere.  (\textbf{A2b}) The network $\mathcal{D}_{\theta}$ is $\beta$-Lipschitz with $\beta<1$ implying  $\| \mathcal{D}_{\theta}(\mathbf{u})-\mathcal{D}_{\theta}(\mathbf{v}) \|_{2}
    \le \beta\|\mathbf{u}-\mathbf{v}\|_{2},$ This condition is commonly adopted and validated empirically in unrolling literature \cite{kofler2023learning, scarlett2023theoretical}
With the step size $0<\tau<(1-\beta)/\beta L$, the unrolled operator is
\begin{equation}
\label{eq:iteration_operator}
    T_{\alpha}(\mathbf{x})
    = \mathcal{D}_{\theta}\!\bigl(\mathbf{x}-\tau\nabla g_{\alpha}(\mathbf{x})\bigr).
\end{equation}

\begin{theorem}[Smooth path of unrolled solutions]
\label{thm:smooth_path_neural}
Under \textbf{A1}–\textbf{A2} the fixed-point equation $\mathbf{x}=T_{\alpha}(\mathbf{x})$ admits a unique solution $\widehat{\mathbf{x}}_{\alpha}=\operatorname*{fix}(T_{\alpha})$ for every $\alpha\in[0,1]$.  The mapping $\alpha\mapsto\widehat{\mathbf{x}}_{\alpha}$ lies in $\mathcal{C}^{1}[0,1]$ and obeys
\vspace{-0.1cm}
\small{\begin{equation}
\vspace{-0.1cm}
\label{eq:lip_bound_theorem}
    \|\widehat{\mathbf{x}}_{\alpha_{1}}-\widehat{\mathbf{x}}_{\alpha_{2}}\|_{2}
    \le \frac{\tau L}{1-\beta(1-\tau L)}\,|\alpha_{1}-\alpha_{2}|,
    \quad \forall\,\alpha_{1},\alpha_{2}\in[0,1].
\end{equation}}
\end{theorem}

\begin{proof}
For any $\mathbf{u},\mathbf{v}$ the $\beta$-Lipschitz property and {Eq.} \eqref{eq:iteration_operator} give
\vspace{-0.1cm}
\begin{align*}
\vspace{-0.1cm}
\|T_{\alpha}(\mathbf{u})-T_{\alpha}(\mathbf{v})\|_{2}
&\le\beta\|\mathbf{u}-\mathbf{v}\|_{2}+\beta\tau L\|\mathbf{u}-\mathbf{v}\|_{2}=\beta(1+\tau L)\|\mathbf{u}-\mathbf{v}\|_{2}.
\end{align*}
Because $\tau<(1-\beta)/L$, the contraction factor $\rho=\beta+\tau L$ satisfies $\rho<1$. Now, let $G(\mathbf{x},\alpha)=\mathbf{x}-T_{\alpha}(\mathbf{x})$.  Since $\mathcal{D}_{\theta}$ and $\nabla g_{\alpha}$ are $\mathcal{C}^{1}$, so is $G$.  The corresponding Jacobian in $\mathbf{x}$ is
\vspace{-0.1cm}
\[
\vspace{-0.1cm}
\mathbf{J}_{\mathbf{x}}G=\mathbf{I}-\mathbf{J}_{\mathcal{D}_{\theta}}\bigl(\mathbf{I}-\tau\nabla^{2}g_{\alpha}\bigr),
\]
whose smallest singular value is at least $1-\beta-\tau L>0$; hence $\mathbf{J}_{\mathbf{x}}G$ is invertible.  The Implicit Function Theorem then yields a unique $\mathcal{C}^{1}$ curve $\alpha\mapsto\widehat{\mathbf{x}}_{\alpha}$ satisfying $G(\widehat{\mathbf{x}}_{\alpha},\alpha)=0$. Then, differentiating $G(\widehat{\mathbf{x}}_{\alpha},\alpha)=0$ with respect to $\alpha$ gives the linear relation $\mathbf{J}_{\mathbf{x}}G\,\dot{\widehat{\mathbf{x}}}_{\alpha}+\partial_{\alpha}G=0$.  Taking norms and using $\|\mathbf{J}_{\mathbf{x}}G^{-1}\|_{2}\le(1-\beta(1-\tau L))^{-1}$ and $\|\partial_{\alpha}G\|_{2}\le\tau L$ yields $\|\dot{\widehat{\mathbf{x}}}_{\alpha}\|_{2}\le\tau L/(1-\beta(1-\tau L))$.  Finally, for any \(\alpha_1,\alpha_2\in[0,1]\) apply the Fundamental Theorem of
Calculus: \vspace{-0.25cm}
\begin{align*}
     \|\widehat{\mathbf{x}}_{\alpha_1}-\widehat{\mathbf{x}}_{\alpha_2}\|_2
   &=\Bigl\|\int_{\alpha_2}^{\alpha_1}\dot{\widehat{\mathbf{x}}}_\tau\,d\alpha\Bigr\|_2
   \le
   \int_{\alpha_2}^{\alpha_1}\frac{\tau L}{1-\beta(1-\tau L)}d\alpha = 
   \frac{\tau L}{1-\beta(1-\tau L)}\,|\alpha_1-\alpha_2|.
\end{align*}

\noindent leads directly to the global bound {Eq.} \eqref{eq:lip_bound_theorem}.
\end{proof}

\begin{remark}The results imply that the difference between the unrolling solutions along the continuation path is bounded up to a constant of the difference between the value of $\alpha$.
\end{remark}

\begin{figure*}[!t]
    \centering
    \includegraphics[width=\linewidth]{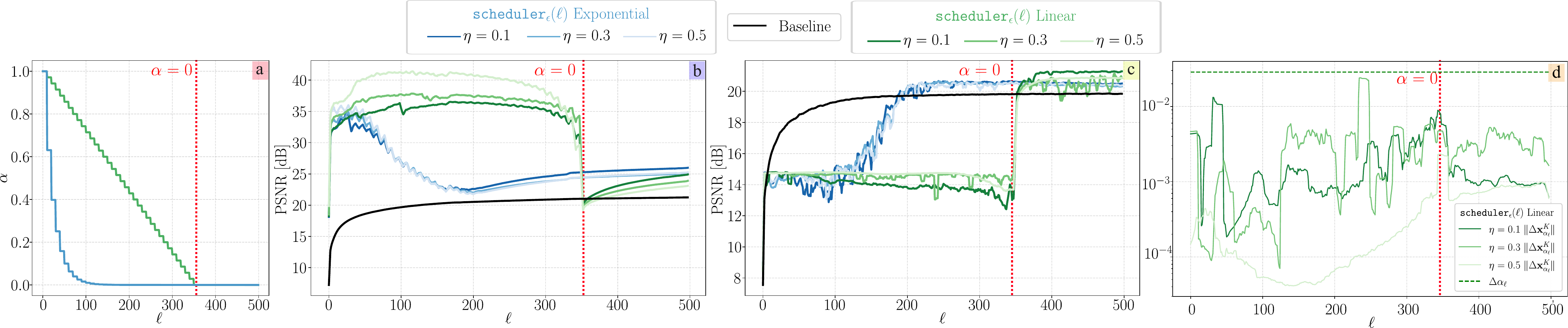}
    \vspace{-0.6cm}
    \caption{(a) Exponential and linear homotopy‐continuation schedules $\alpha(\ell)$ over training epochs, updating the $\alpha$ every 10 epochs until reach $\alpha=0$ at epoch 350. (b) Training PSNR per epoch for each scheduler and $\eta$, compared against the baseline unrolled network trained with constant $\alpha=0$ for the CS experiment. (c) Testing PSNR on the target fidelity $(\mathbf{H},\mathbf{y})$: although the baseline converges fastest, the proposed homotopy method surpasses the baseline’s final PSNR in all cases and for both schedulers. (d) Experimental validation of Theorem \ref{thm:smooth_path_neural}.}
    \label{fig:UTOPY_SchTrainTest}
    \vspace{-0.65cm}
\end{figure*}

\section{Results}\label{sec:results}

 The framework was implemented in PyTorch. For the network training, we used the CelebA dataset \cite{liu2015faceattributes}, using 24,318 images for training and 2,993 for testing. All images were resized to $64\times 64$, thus $n=4096$ and converted to grayscale. During training, additive white Gaussian noise was added to the measurements with SNR = 35[dB]. We used the Adam optimizer \cite{adam} with a learning rate $\zeta = 1e^{-5}$, which was halved three times during the network training. We used a composite loss function 
 \begin{align*}\mathcal{L}_{\texttt{rec}} (\mathbf{a},\mathbf{b}) = 0.8\times\Vert\mathbf{a-b}\Vert_1 + 0.2 \times (1-\texttt{SSIM}(\mathbf{a,b})) + 0.02\times \Vert\mathbf{w} \odot (\vert \mathbf{Fa-Fb} \vert) \Vert,\end{align*} where $\mathbf{F}$ is the Fourier transform and $\mathbf{w}$ is a circular window, band passing the high frequencies. This loss function is inspired by the focal frequency loss \cite{jiang2021focal} and the loss function proposed in \cite{loss}. The weights of each term were selected by grid search on the unrolling baseline training. For the implementation, we followed the library \texttt{pycolibri} \cite{colibri}. The baseline consisted of the classical training of unrolling networks, i.e., $\alpha=0 \forall \ell$.  The network $\mathcal{D}_{\theta^k}$
was a UNet model \cite{unet} with features following $[32,64,128,256]$, and each level has a double block convolution, batch normalization, and ReLU. We used $K=5$ unrolling stages, the parameters were initialized $\tau^k = 1e^{-3}\forall k$. The acceleration parameter $t^k$ of Algorithm \ref{alg:fidelity_optimization}  is trainable. For the linear scheduler $\epsilon = \frac{1}{0.7\texttt{max\_epochs}}$ and for the exponential case we set $\epsilon = \log\left(\frac{1e^{-8}}{0.7\texttt{max\_epochs}}\right)$. The value 0.7 was set such that we achieve $\alpha=0$ at 70\% of the training so that the last epochs focus on the target inverse problem $(\mathbf{H,y})$.

\begin{figure}[!t]
    \centering
    \includegraphics[width=0.6\linewidth]{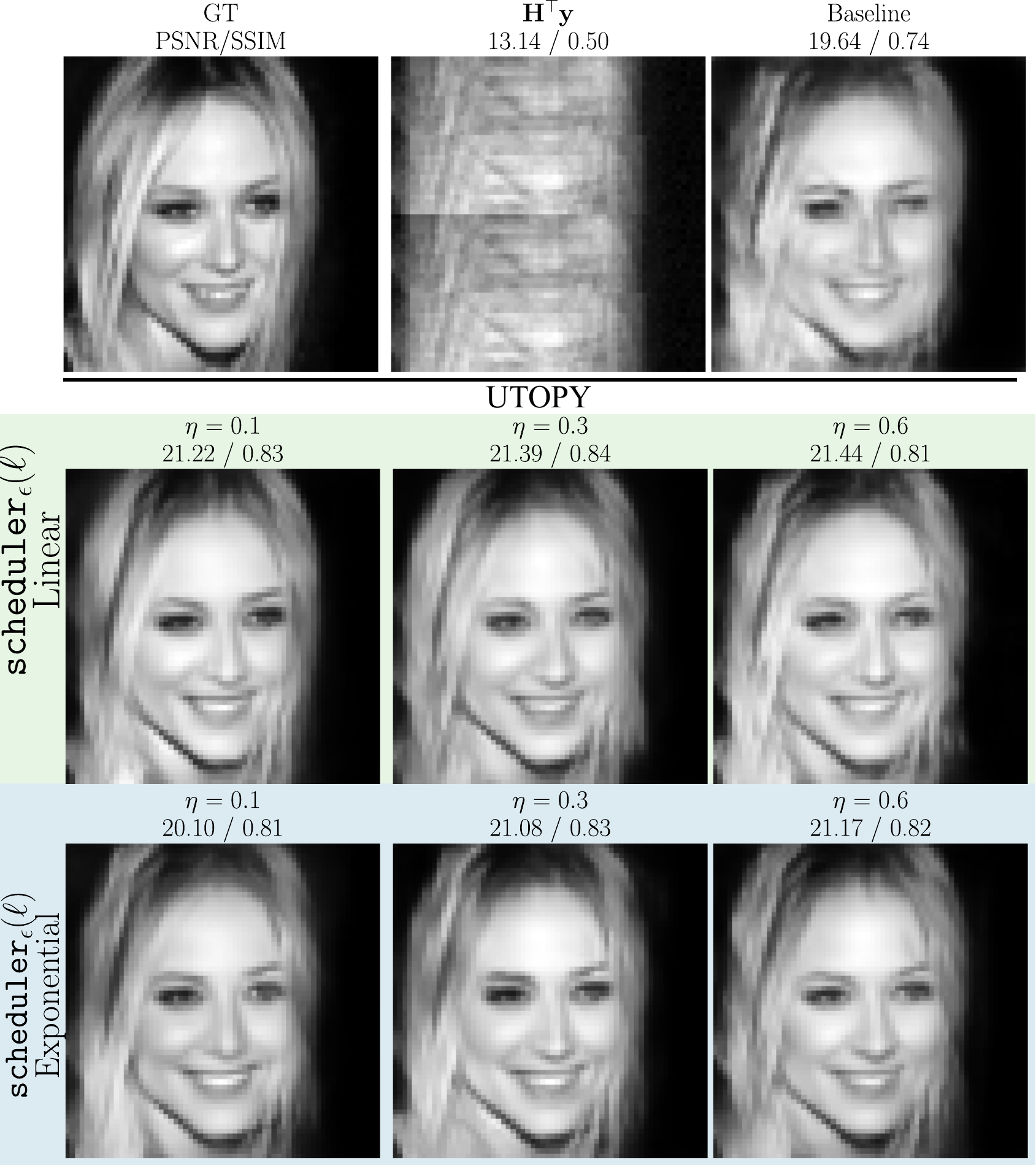}
    \caption{Visual reconstruction for SPC experiment with different schedulers and the baseline unrolling model.}
    \label{fig:visual_spc}
\end{figure}

 \noindent\textbf{Compressed sensing: } To validate the proposed approach in single-pixel imaging \cite{hu2019single,duarte2008single} using Hadamard coded apertures using the implementation of \cite{monroy2024hadamard}.  We set the compression ratio for all experiments, i.e., $m/n = 0.3$. We used  $\texttt{max\_epochs} = 500$ and a batch size of 32. Fig. \ref{fig:UTOPY_SchTrainTest}(a) shows both schedulers. We set $\texttt{freq} = 10$. Fig. \ref{fig:UTOPY_SchTrainTest}(b) reports the averaged training PSNR: early in training, all homotopy variants quickly gain PSNR by solving the simpler $(\mathbf{H}_t,\mathbf{y}_t)$ task. As $\alpha \to 0$, the PSNR temporarily declines, which reflects the shift to the ill‐posed fidelity, before recovering, demonstrating effective transition. The baseline (constant $\alpha=0$) shows faster convergence but lower performance. Fig. \ref{fig:UTOPY_SchTrainTest}(c) presents the testing PSNR on true measurements $\mathbf{y}$: the baseline converges faster due to its constant focus on the target fidelity, but the proposed homotopy‐trained models achieve higher final PSNR across all $\eta$ values and both scheduler types, confirming that homotopy fidelity continuation yields superior generalization in the inverse problem. Additionally, we validate experimentally the claims in Theorem \ref{thm:smooth_path_neural}. In which dotted tile $\Delta \alpha_\ell = |\alpha_{\ell}-\alpha_{\ell-1}|$ is the difference between the changes in $\alpha$ (in this case a constant as the scheduler was linear), and $ \Vert\Delta \mathbf{x}_\alpha ^\ell\Vert = \Vert\mathbf{x}_\alpha ^\ell-\mathbf{x}_\alpha ^{\ell-1}\Vert$ are the changes in the solutions of the unrolling algorithm. Figure \ref{fig:UTOPY_SchTrainTest}d) shows that along the network training, the changes in the solutions are bounded by the changes in the $\alpha$ parameter.  In Fig. \ref{fig:visual_spc}, the proposed method is shown with both schedulers at different settings of $\eta$, showing better reconstructions than base unrolling training.
\begin{figure}[!t]
    \centering
    \includegraphics[width=0.6\linewidth]{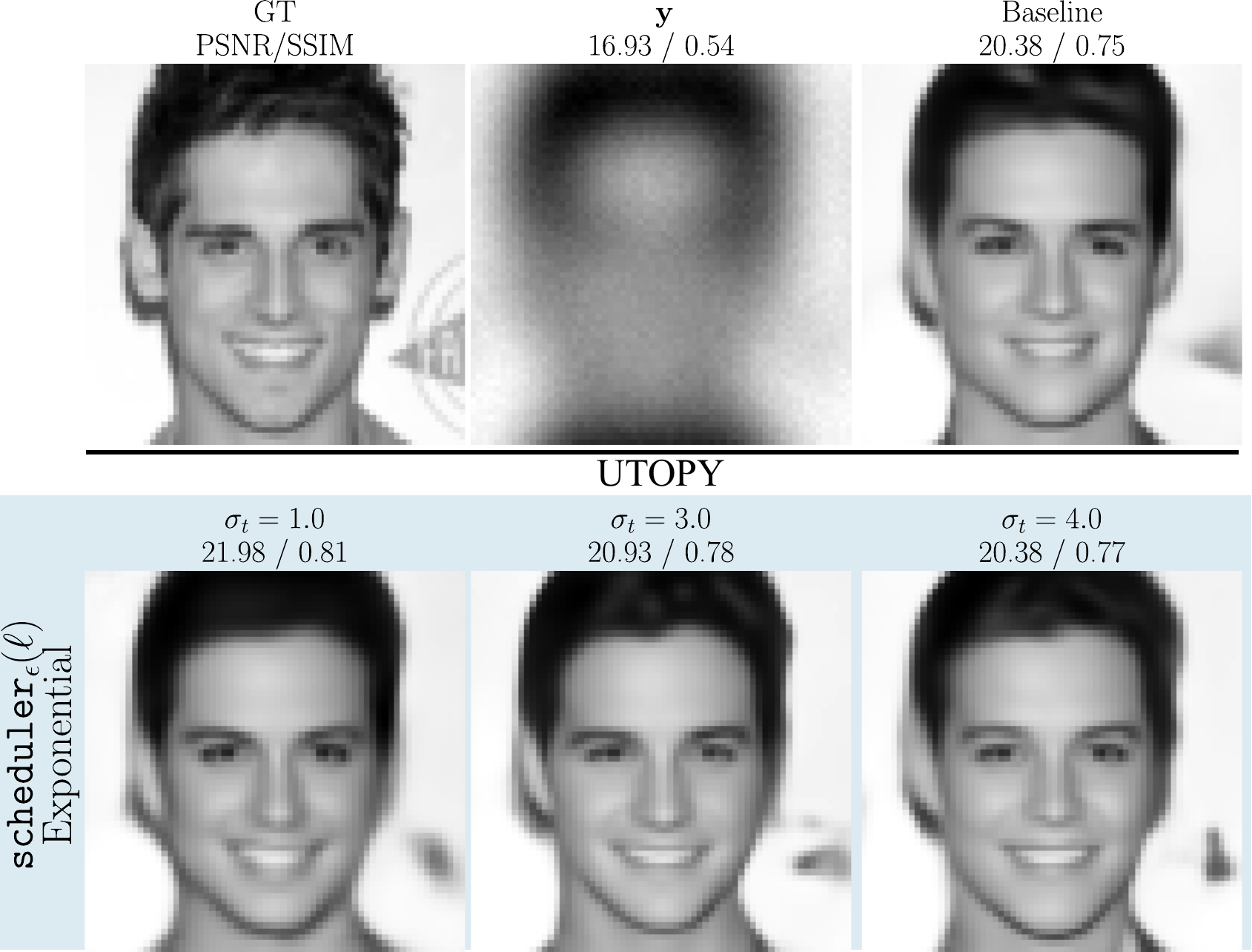}
    \caption{Visual reconstruction for image deblurring using exponential UTOPY and the baseline unrolling model.}
    \label{fig:visual_deblur}    
\end{figure}
\noindent\textbf{Image deblurring:} We set for these experiments $\sigma=5$, used a batch size of 100, and $\texttt{max\_epochs} = 200$. For the setting of UTOPY, we set $\sigma_t = [1,3,4]$ for the synthetic well-posed sensing matrix $\mathbf{H}_t$. For proof of concept of the method in this task, we only used the exponential scheduler. The results in Figure \ref{fig:visual_deblur} show the reconstruction performance for the proposed method and baseline unrolling. Although the measurements are highly degraded, the proposed method achieves high-quality reconstruction up to 1.6 dB.

\noindent\textbf{Generalization of changes in signal model:} Another advantage of the proposed approach is that, since it uses different fidelity functions along the continuation path, it resembles to training with several sensing matrices, achieving generalization to changes in the sensing model beyond $(\mathbf{H,y})$. We change the compression ratio in the CS experiment with $[0.25,0.35]$ and the SNR $[30, 40]$ in the deblurring experiment. In Table \ref{tb:table_generaliztion}, the performance for the CelebA test set is shown. For UTOPY models,  we set the one with the linear scheduler and $\eta=0.1$, and for deblurring, the one with $\sigma_t = 1$. We also compared with the FISTA plug-and-play algorithm with DnCNN denoiser \cite{kamilov2023plug,zhang2021plug} using the trained model in \texttt{DeepInverse} library \cite{tachella2025deepinverse}. The results show that UTOPY delivers consistent improvements across all testing scenarios, showing the generalization advantage of the homotopy training.

%

\begin{table}[!t]
\centering
\resizebox{0.8\linewidth}{!}{\begin{tabular}{cccccccc}
\toprule
\multirow{2}{*}{\textbf{Inverse Problem}}
  & \multirow{2}{*}{\textbf{Testing Setting}}
  & \multicolumn{2}{c}{\textbf{Baseline}}
  & \multicolumn{2}{c}{\textbf{UTOPY}}  & \multicolumn{2}{c}{\textbf{FISTA-PnP}}\\
\cmidrule(lr){3-4}\cmidrule(lr){5-6}\cmidrule(lr){7-8}
  & & \textbf{PSNR} & \textbf{SSIM} & \textbf{PSNR} & \textbf{SSIM}& \textbf{PSNR} & \textbf{SSIM} \\
\midrule
\multirow{3}{*}{CS}
  & $m/n = 0.25$ & 14.97 & 0.525 & \colorbox{customteal}{\textbf{16.78}} & \colorbox{customteal}{\textbf{0.620 }}  & 14.55 &  0.525\\
  & \colorbox{customorange}{$m/n = 0.30$} & 19.75 & 0.706 & \colorbox{customteal}{\textbf{21.19}} & \colorbox{customteal}{\textbf{0.771}}&15.20& 0.493\\
  & $m/n = 0.35$ & 19.93 & 0.720 & \colorbox{customteal}{\textbf{21.40}} & \colorbox{customteal}{\textbf{0.784 }}&15.37&0.510\\
\midrule
\multirow{3}{*}{Deblurring}
  & SNR = 30     & 20.84 & 0.658 & \colorbox{customteal}{\textbf{21.54}} & \colorbox{customteal}{\textbf{0.693 }}& 19.85 & 0.596\\
  &  \colorbox{customorange}{SNR = 35}     & 21.17 & 0.660 & \colorbox{customteal}{\textbf{22.17}} & \colorbox{customteal}{\textbf{0.718}} & 19.906 & 0.598\\
  & SNR = 40     & 21.12 & 0.668 & \colorbox{customteal}{\textbf{22.38}} & \colorbox{customteal}{\textbf{0.726}} & 19.92&0.598\\
\bottomrule
\end{tabular}}
\vspace{-0.1cm}
\caption{Unrolling baseline, UTOPY, and FISTA-PnP performance across testing settings for each inverse problem.  \colorbox{customorange}{Setting used in training}. Best performance for each testing setting in \colorbox{customteal}{\textbf{bold teal}.} }\vspace{-0.6cm}
\label{tb:table_generaliztion}
\end{table}

\vspace{-0.1cm}
\section{Conclusion}
\vspace{-0.1cm}
\noindent We proposed UTOPY, a new approach to train unrolling networks that harness homotopy continuation path strategies to learn the reconstruction model from a well-posed (synthetic) inverse problem and smoothly transition towards the challenging target inverse problem. We show theoretically that the possible unrolling solutions have smooth transitions along the homotopy continuation path. The proposed method improves reconstruction performance and generalization upon traditional unrolling training and plug-and-play methods for both compressed sensing and image deblurring. 

\bibliographystyle{IEEEtran}
\bibliography{bibliography}

\end{document}